\documentclass[letterpaper, 10 pt, journal, twoside]{IEEEtran}  

\IEEEoverridecommandlockouts                              

\usepackage{graphics} 
\usepackage{epsfig} 
\usepackage{amsmath} 
\usepackage{amssymb}  

\usepackage{physics}
\usepackage{algpseudocode}
\usepackage{epstopdf}
\usepackage{array}

\usepackage{verbatim}

\usepackage{color}
\usepackage{cancel}

\usepackage{tikz}               
\usepackage{subfigure}

\usepackage{setspace}
\usepackage{amsfonts}
\usepackage{cuted}

\newtheorem{theorem}{Theorem}

\newtheorem{lemma}[theorem]{Lemma}

\newtheorem{remark}[theorem]{Remark}

\newtheorem{assumption}[theorem]{Assumption}

\newcommand{\change}[2]{\textcolor{red!80}{}{#2}}

\title{Learning an Approximate Model Predictive Controller with Guarantees }

\author{Michael Hertneck$^1$, Johannes K\"ohler$^{2}$, Sebastian Trimpe$^{3}$, Frank Allg\"ower$^{2}$
\thanks{$^1$Michael Hertneck is an M.Sc. student at the University of Stuttgart, 70550 Stuttgart, Germany (email: michaelhertneck@yahoo.de).}%
\thanks{$^2$Johannes K\"ohler and Frank Allg\"ower are with the Institute for Systems Theory and Automatic Control, University of Stuttgart, 70550 Stuttgart, Germany (email: $\{$johannes.koehler,  frank.allgower$\}$@ist.uni-stuttgart.de).}%
\thanks{$^3$Sebastian Trimpe is with the Intelligent Control Systems Group at the Max Planck Institute for Intelligent Systems, 70569 Stuttgart, Germany (email: trimpe@is.mpg.de). }%
\thanks{This work was supported in part by the German Research Foundation (DFG) grant GRK 2198/1, the Max Planck Society, and the Cyber Valley Initiative. }%
 
}

\usepackage{fancyhdr}
\newcommand{\mytitle}{\textbf{Accepted version.}
To appear in \textit{IEEE Control Systems Letters}.  DOI: 10.1109/LCSYS.2018.2843682\\
\copyright 2018 IEEE. Personal use of this material is permitted. Permission from IEEE must be obtained for all other uses, in any current or future media, including reprinting/republishing this material for advertising or promotional purposes, creating new collective works, for resale or redistribution to servers or lists, or reuse of any copyrighted component of this work in other works.}
\begin{document}

\maketitle

\thispagestyle{fancy}	
\pagestyle{empty}


\begin{abstract}
A supervised learning framework is proposed to approximate a model predictive controller (MPC) with reduced computational complexity and guarantees on stability and constraint satisfaction. The framework can be used for a wide class of nonlinear systems. Any standard supervised learning technique (e.g. neural networks) can be employed to approximate the MPC from samples. In order to obtain closed-loop guarantees for the learned MPC, a robust MPC design is combined with statistical learning bounds. The MPC design ensures robustness to inaccurate inputs within given bounds, and Hoeffding's Inequality is used to validate that the learned MPC satisfies these bounds with high confidence. The result is a closed-loop statistical guarantee on stability and constraint satisfaction for the learned MPC. The proposed learning-based MPC framework is illustrated on a nonlinear benchmark problem, for which we learn a neural network controller with guarantees.

\end{abstract}
\section{Introduction}
\IEEEPARstart{M}{odel} predictive control (MPC) \cite{rawlings2009model} is a modern control method based on repeatedly solving an optimization problem online. It can handle general nonlinear dynamics, hard state and input constraints, and general objective functions. One major drawback of MPC is the computational effort of solving optimization problems online under real-time requirements. Especially for settings with a large number of optimization variables or if a high sampling rate is required, the online optimization may get computationally intractable. Hence, it is often desirable to find an explicit formulation of the MPC that can be evaluated online in a short deterministic execution time, also on relatively inexpensive hardware. 

For linear systems, the MPC optimization problem can be formulated as a multi-parametric quadratic program, which can be solved offline to obtain an explicit control law \cite{bemporad2002explicit}. The extension of \cite{bemporad2002explicit} to nonlinear systems is not straightforward. Hence, the goal of this paper is to develop a framework for approximating a nonlinear MPC through supervised learning with statistical guarantees on stability and constraint satisfaction.

A sketch of the main ideas is given in Figure \ref{fig_diagram2}.
First, a \emph{robust} MPC (RMPC) design is carried out that ensures robustness to bounded input disturbances $ d $ with a user defined bound ($| d | \leq \eta$).  The resulting RMPC feedback law $\pi_\text{\normalfont MPC}(x)$ is sampled offline for suitable states $x_i$ and approximated via any function approximation (regression) technique such as neural networks (NNs).  If the error of the approximate control law $\pi_\text{\normalfont approx}(x)$ is below the admissible bound of the RMPC, recursive feasibility and closed-loop stability for the learned RMPC can be guaranteed.  In order to guarantee a sufficiently small approximation error, we use Hoeffding's Inequality on a suitable validation data set.  The overall result is an approximate MPC (AMPC) with lower computational requirements 
and statistical guarantees on closed-loop stability and constraint satisfaction. 
The proposed approach is applicable to a wide class of nonlinear control problems with state and input constraints.

\usetikzlibrary{backgrounds}
\usetikzlibrary{shapes,arrows,chains}
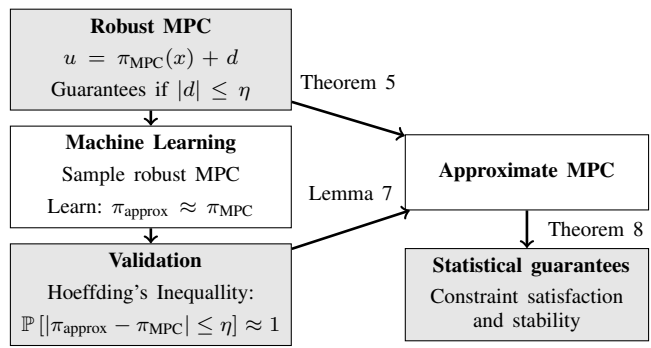
\begin{figure}
		\begin{tikzpicture}
	
   	\node[draw, text width=3.5cm,align=center,fill = gray!20] (node1) {\footnotesize\textbf{Robust MPC} \\ $u = \pi_\text{\normalfont MPC}(x) +  d $ Guarantees if $\abs{ d } \leq \eta$};
   	\node[draw,below  = 0.2 cm  of node1, text width=3.5cm,align=center, minimum height=1.1cm] (node3) {\footnotesize \textbf{Machine Learning} \\
    Sample robust MPC 
    Learn:  $\pi_\text{\normalfont approx} \approx \pi_\text{\normalfont MPC}$ 
   		};
   	\node[draw,below  = 0.2 cm  of node3, text width=3.5cm,align=center, minimum height=1.1cm,fill = gray!20] (node4) {\footnotesize \textbf{ Validation } \\ Hoeffding's Inequallity: $\mathbb{P}\left[\abs{ \pi_\text{\normalfont approx} - \pi_\text{\normalfont MPC}} \leq \eta \right]\approx 1$};
   	\node[draw, text width=3cm, minimum height=1cm, align=center] (node5) at (5,-1.5) {\footnotesize \textbf{Approximate MPC} };
   	\node[draw,below  = 0.5 cm  of node5, text width=3cm,align=center,fill = gray!20] (node6) {\footnotesize  \textbf{ Statistical guarantees}\\[0.8ex]   Constraint satisfaction\\[-0.8ex] and stability};
   	\draw [->,line width = 1pt] (node1) -- (node3);

   	\draw [->,line width = 1pt] (node3) -- (node4);
   	\draw [->,line width = 1pt] (node1) -- (node5);
   	\draw [->,line width = 1pt] (node4) -- (node5);
   	\draw [->,line width = 1pt] (node5) -- (node6);
   	\path [->,line width = 1pt] (node1) -- node [text width=2.5cm,midway,above = 0.75em,align = center] {\footnotesize Theorem \ref{theo_robust}} (node5);
   	\path [->] (node4) -- node [text width=2.5cm,midway,above = 0.75em,align = center] {\footnotesize Lemma \ref{theo_hoeff}} (node5);
   		\path [->] (node5) -- node [text width=2.5cm,midway,right = -1.3 em,align = center] {\footnotesize Theorem \ref{lem_algo}} (node6);
   	
   	\end{tikzpicture}
	\vspace{-0.5 cm}
	\caption{Diagram of the proposed framework. We design an MPC $\pi_\text{MPC}$ with robustness to input disturbance $d$. The resulting feedback law is sampled offline  and approximated ($\pi_\text{approx}$) via supervised learning. Hoeffding's Inequality is used for validation and yields a bound on the error between approximate and original MPC in order to guarantee stability and constraint satisfaction. The result is an approximate MPC
		with statistical guarantees.} 
	\label{fig_diagram2}
\end{figure}

\paragraph*{Contributions}
This paper makes the following contributions:
We present an MPC design which is robust for a \textit{chosen} bound on the input error. Furthermore, we propose a validation method based on Hoeffding's Inequality to provide a statistical bound on the approximation error of the approximated MPC. By combining these approaches, we obtain a complete framework to learn an AMPC from offline samples in an automatic fashion. The framework is suited for nonlinear systems with constraints, can incorporate a user defined cost function, and supports high sampling rates on cheap hardware.  The framework is demonstrated on a numerical benchmark example, where the AMPC is represented by an NN.

\paragraph*{Related work} 
In the literature, there exist several approaches to obtain offline an approximate solution for an MPC.
\change{}{In~\cite{goebel2017semi}, a semi-explicit MPC for linear systems is presented that ensures stability and constraint satisfaction with a decreased online
	computational demand, by combining a subspace clustering with a
	feasibility refinement strategy.}
In \cite{domahidi2011learning}, a learning algorithm is presented with additional constraints to guarantee stability and constraint satisfaction of the approximate MPC for linear systems. 
\change{This stands}{Both approaches stand} in contrast to our approach, which can be used with any learning method and supports nonlinear systems.

One approach to approximate a nonlinear MPC is convex multi-parametric nonlinear programming \cite{johansen2004approximate,grancharova2009computation}. Approximating an MPC by NNs, as also done herein, has for example been proposed in \cite{parisini1998nonlinear,aakesson2006neural}. In contrast to the proposed framework, these approaches cannot \textit{guarantee} stability and constraint satisfaction for the resulting AMPC.
In \change{\cite{chakrabarty2017support}}{\cite{chakrabarty2017support,canale2009fast}}, it is shown, that guarantees on stability and constraint satisfaction are preserved for arbitrary small approximation errors (due to inherent robustness properties). 
In \cite{pin2013approximate}, an MPC with Lipschitz based constraint tightening is learned that provides guarantees for a non vanishing approximation error. 
The admissible approximation error deduced in \change{\cite{chakrabarty2017support,pin2013approximate}}{\cite{chakrabarty2017support,canale2009fast,pin2013approximate}} is typically not achievable (compare example) and is thus not suited within proposed framework. 

Neural networks have recently been increasingly popular as control policies for complex tasks, \cite{arulkumaran2017brief}.  One powerful framework for training (deep) NN policies is guided policy search (GPS) \cite{LeFiDaAb16},  which uses trajectory-centric control to generate samples for training the NN policy via supervised learning.  In \cite{zhang2016learning}, an MPC is used with GPS for generating the training samples.  While that work is conceptually similar, it does not provide closed-loop guarantees for the learned controller as we do herein.
The main novelty of the proposed framework lies in the approximation of a \textit{robust} control technique and its combination with a suitable validation approach. The validation method has resemblances to \cite{tempo1997probabilistic}, where statistical guarantees for robustness are derived based on the Chernoff bound.

\paragraph*{Outline} This paper is structured as follows: We formulate the problem and present our main idea in Section \ref{sec_approach}. In Section \ref{sec_rob}, the RMPC design is presented. In Section \ref{sec_learning}, we present a validation method  and propose a procedure to learn an MPC with guarantees. A numerical example to show the applicability of the proposed framework is given in Section \ref{sec_example}. Section \ref{sec_conclusion} concludes the paper.

\paragraph*{Notation} \change{}{
	The euclidean norm with respect to a positive definite matrix $Q=Q^\top$ is denoted by $\|x\|_Q^2=x^\top Q x$.} The Pontryagin set difference is defined by $ X \ominus Y := \left\lbrace z \in \mathbb{R}^n: z+y\in X, \forall y \in Y \right\rbrace $. For $X$ a random variable,  $\mathbb{E}(X)$ denotes the expectation of $X$, and $\mathbb{P}\left[X \geq x\right]$ denotes the probability for $X \geq x$. \change{}{We denote $\mathbf{1}_p=[1,..,1]^\top\in\mathbb{R}^p$.
	}.

\section{Main approach}
\label{sec_approach}
In this section, we pose the control problem and describe the proposed approach. 
\subsection{Problem formulation}
We consider the following nonlinear discrete-time system
\begin{equation}
\label{eq_sys}
x(t+1) = f(x(t),u(t))
\end{equation}
with the state $x(t) \in \mathbb{R}^n$, the control input $u(t) \in \mathbb{R}^m$, the time step $t \in \mathbb{N}$, and $f$ continuous with $f(0,0) = 0$. We consider compact polytopic constraints 
\begin{equation}
\nonumber
\mathcal{X} = \left\lbrace x \in \mathbb{R}^n|Hx \leq \textbf{1}_p\right\rbrace,~ \mathcal{U} = \left\lbrace u \in \mathbb{R}^m|Lu \leq \textbf{1}_q\right\rbrace.
\end{equation} 
The control objective is to ensure stability of $x = 0$, constraint satisfaction, i.e. $(x(t),u(t)) \in \mathcal{X} \times \mathcal{U}~\forall t \geq 0$, and optimize some cost function. We shall consider these objectives under initial conditions $x(0) \in \mathcal{X}_\text{feas}$, where $\mathcal{X}_\text{feas}$ is a feasible set to be made precise later.
The resulting controller should be implementable on cheap hardware for systems with high sampling rates.

\subsection{General approach}

The controller synthesis with the proposed framework works as follows: An MPC is designed for \eqref{eq_sys} that is robust to inaccurate inputs $u$ within chosen bounds. The RMPC is sampled offline over the set of feasible states $\mathcal{X}_\text{\normalfont feas}$ and approximated using supervised learning techniques based from these samples. In this paper, we will use NNs to approximate the RMPC, but any other supervised learning technique \change{}{or regression method} can be used likewise.
The learning yields an AMPC $\pi_\text{\normalfont approx}: \mathcal{X}_\text{\normalfont feas} \rightarrow \mathcal{U}~| u =  \pi_\text{\normalfont approx}(x)$.
 With this controller, the closed-loop system is given by
\begin{equation}
\label{eq_closed_loop}
x(t+1) = f(x(t),\pi_\text{\normalfont approx}(x(t))).
\end{equation}
Stability of the closed loop \eqref{eq_closed_loop} is guaranteed if the approximation error is below the admissible bound on the input disturbance from the RMPC. We use a validation method based on Hoeffding's Inequality to guarantee this bound.



\section{Input Robust MPC}
\label{sec_rob}
In this section, we present an input robust MPC design with robust guarantees on stability and constraint satisfaction for bounded additive input disturbances. To achieve the robustness, we combine a standard MPC formulation~\cite{chen1998quasi} with a robust constraint tightening \cite{koehler2017novel}. The RMPC optimization problem can be formulated as   
\begin{subequations}
	\label{eq_opt}
		\begin{align}
	\label{eq_qihmpc}
	\underset{u(\cdot|t)}{\min }&\sum_{k = 0}^{N-1} \norm{x(k|t)}^2_Q + \norm{u(k|t)}^2_R + \norm{x(N|t)}^2_P\\
	s. t.~&
	x(0|t) = x(t),\quad x(N|t) \in  \mathcal{X}_f, \\	
	&x(k+1|t) = f(x(k|t),u(k|t)),	\label{eq_mpc_constr} \\
	&x(k|t) \in \bar{\mathcal{X}}_k, \quad u(k|t) \in  \bar{\mathcal{U}}_k,~ k =0,...,N-1	
	\end{align}
\end{subequations}
We denote the set of states $x$ where \eqref{eq_opt} is feasible by $\mathcal{X}_\text{\normalfont feas}$. The solution to the RMPC optimization problem~\eqref{eq_opt} is denoted by $u^*(\cdot|t)$. The MPC feedback law is $\pi_\text{\normalfont MPC}(x(t)):=u^*(0|t)$. The state and input constraints from standard MPC formulations are replaced by tightened constraints $ \bar{\mathcal{X}}_k$ and $\bar{\mathcal{U}}_k $. This tightening guarantees stability and constraint satisfaction despite bounded input errors and 
will be discussed in more detail in the following. The positive definite matrices $Q$ and $R$ are design parameters. The design of the terminal ingredients $\mathcal{X}_f$ and $P$ will be made precise later. 
 The closed-loop of the RMPC under input disturbances is given by 
 \begin{equation}
 \label{eq_cl_rob}
 x(t+1) = f(x(t),\pi_\text{\normalfont MPC}(x(t)) + d(t))
 \end{equation}
 with $d(t) \in \mathcal{W} = \left\lbrace d\in\mathbb{R}^m : \norm{d}_\infty \leq \eta \right\rbrace,~\forall t \geq 0$ for some $\eta$. The following assumption will be used in order to design the RMPC:
\begin{assumption}
	\label{loc_inc_stab}
	(Local incremental stabilizability \cite{koehler2017novel,koehler201?reference}) 	
	There exists a control law $\kappa: \mathcal{X}\times \mathcal{X} \times \mathcal{U}\rightarrow\mathbb{R}^m$,\linebreak a $\delta$-Lyapunov function $V_{\delta}: \mathcal{X}\times \mathcal{X} \times\mathcal{U}\rightarrow\mathbb{R}_{\geq 0 } $, that is continuous in the first argument and satisfies \linebreak $V_\delta (x,x,v)=0 ~\forall x\in\mathcal{X},~\forall v\in \mathcal{U} $, and parameters $c_{\delta,l},$ $ c_{\delta,u}$ $ \delta_\text{\normalfont loc}, $ $k_{\max} \in \mathbb{R}_{>0}, \rho\in (0,1)$, such that the following properties hold for all $(x,z,v) \in \mathcal{X} \times \mathcal{X} \times \mathcal{U}, (z^+,v^+) \in \mathcal{X} \times \mathcal{U}$ with $V_\delta(x,z,v) \leq \delta_\text{\normalfont loc} $: 
	\begin{align}
	c_{\delta,l} \norm{x-z}^2 \leq V_\delta(x,z,v) &\leq c_{\delta,u} \norm{x-z}^2, \label{eq_Vdel} \\	
	\norm{\kappa(x,z,v)-v} &\leq k_{\max} \norm{x-z}, \label{eq_kmax} \\	
	V_\delta(x^+,z^+,v^+) &\leq \rho V_\delta (x,z,v) \label{eq_contr} 
	\end{align}
	with $x^+ = f(x,\kappa (x,z,v)),~z^+= f(z,v)$.

\end{assumption}
\change{Assumption \ref{loc_inc_stab} is, for example, fulfilled if the linearized dynamics at any point $(x,v) \in \mathcal{X} \times \mathcal{U}$ are stabilizable with a common quadratic Lyapunov function $V = x^\top P x$ and $f$ is locally Lipschitz. Hence, Assumption~\ref{loc_inc_stab} is not very restrictive.}{This assumption is quite general. Sufficient conditions for this property can be formulated based on the linearization, provided that $f$ is locally Lipschitz, compare \cite{koehler201?reference}.}  
The concept of incremental stability \cite{koehler2017novel} describes an incremental robustness property, and is thus suited for the constraint tightening along the prediction horizon. To overestimate the influence from the system input on the system state, we use the following assumption:
\begin{assumption}
		\label{as_lip}
	(Local Lipschitz continuity) There exists a $\lambda \in \mathbb{R}$, such that $\forall x \in \mathcal{X}, \forall u \in \mathcal{U}, \forall  u+d \in \mathcal{U}$  
	\begin{equation}
	\norm{f(x,u+ d ) - f(x,u)} \leq \lambda \norm{ d }_\infty.
	\end{equation}
\end{assumption}
With this assumption, we can introduce a bound on the admissible input disturbance which will be used in the proof of robust stability and recursive feasibility for the RMPC:
\begin{assumption} 
	\label{as_umax}
(Bound on the input disturbance) The input disturbance bound satisfies
  $\eta \leq  \frac{1}{\lambda} \sqrt{\frac{\delta_\text{\normalfont loc}}{c_{\delta, u}}}.$
\end{assumption}
To guarantee robust constraint satisfaction we use a growing tube inspired constraint tightening, based on the incremental stabilizability property in Assumption \ref{loc_inc_stab} as in \cite{koehler2017novel}.
Consider the polytopic tightened set $\mathcal{U}_t = \mathcal{U} \ominus \mathcal{W}  \linebreak =\left\lbrace u \in \mathbb{R}^m: L_t u \leq \textbf{1}_p\right\rbrace$.
 This set ensures that $u+d \in \mathcal{U}, ~\forall u\in \mathcal{U}_t,~\forall d \in \mathcal{W}$. 
We set
\begin{equation}
\label{eq_eps_bound}
\epsilon := \eta \lambda \sqrt{\frac{c_{\delta,u}}{c_{\delta,l}}}\max\left\lbrace \norm{H}_{\infty}, \norm{L_t}_{\infty} k_{\max} \right\rbrace.
\end{equation} The constraint tightening is achieved with a scalar tightening parameter $\epsilon_k := \epsilon \frac{1-\sqrt{\rho}^k}{1-\sqrt{\rho}}, ~ k \in \{ 0, ..., N\}$ based on the asymptotic decay rate $\rho$ and  $\epsilon$.
The tightened constraint sets are given by
\begin{equation}
\nonumber
\bar{\mathcal{X}}_k := (1-\epsilon_k) \mathcal{X} = \{x \in \mathbb{R}^n : Hx \leq (1-\epsilon_k) \textbf{1}_p  \},
\end{equation}
\begin{equation}
\nonumber
\bar{\mathcal{U}}_k := (1-\epsilon_k) \mathcal{U}_t  = \{u \in \mathbb{R}^m : L_tu \leq (1-\epsilon_k) \textbf{1}_q  \}.
\end{equation}
  \change{}{This constraint tightening can be thought of as an over approximation of the constraint tightening used in \cite{chisci2001systems}, with the difference that it can be easily applied to nonlinear systems, compare \cite{koehler2017novel}}. Clearly, the size of the tightened constraints depends on $\epsilon$ and thus on the size of the bound on the input disturbance~$\eta$. The maximum influence of the input disturbance $d$ on the predicted state $x(N|t)$ is bounded by \linebreak $
\mathcal{W}_N = \left\lbrace x \in \mathbb{R}^n : \norm{x} \leq  \lambda \eta \sqrt{\rho^{N}\frac{c_{\delta,u}}{c_{\delta, l}}}  \right\rbrace.
$
 We use the following Assumption on the terminal set to guarantee recursive feasibility and closed-loop stability similar to \cite{chen1998quasi}:
\begin{assumption}
	\label{as_term}
	(Terminal set)  There exists a local control Lyapunov function $V_f(x) = \norm{x}^2_P$, a terminal set $\mathcal{X}_f = \left\lbrace x:V_f(x) \leq \alpha_f \right\rbrace$  and a control law $k_f(x)$, such that$~\forall x \in~\mathcal{X}_f:$
	\begin{align}
	f(x, k_f(x)) + w &\in \mathcal{X}_f,~ \forall w \in \mathcal{W}_N, \label{eq_term_inv} \\
	V_f (f(x,k_f(x))) &\leq V_f(x) - (\norm{x}^2_Q + \norm{k_f(x)}^2_R),	\label{eq_term_constr} \\
	(x,k_f(x)) &\subseteq (\bar{\mathcal{X}}_N  \times \bar{\mathcal{U}}_N). \label{eq_term_sets} 
	\end{align}	
\end{assumption}
It is always possible to design $P,~\mathcal{X}_f$ and $k_f$ to satisfy Assumption~\ref{as_term}, if the linearization of the system is stabilizable, $(0,0)$ lies in the interior of $\mathcal{X} \times \mathcal{U}$ and $\eta$ is small enough.
Now, we are ready to state a theorem that guarantees stability and constraint satisfaction for the RMPC despite bounded input disturbances:
\begin{theorem}
	\label{theo_robust}
	Let Assumption \ref{loc_inc_stab}, \ref{as_lip}, \ref{as_umax}  and \ref{as_term} hold.  Then, for all initial conditions $x(0) \in \mathcal{X}_\text{\normalfont feas}$, the RMPC closed-loop \eqref{eq_cl_rob} satisfies $(x(t),u(t)) \in \mathcal{X}_\text{\normalfont feas} \times \mathcal{U}~ \forall t\geq0$.  Furthermore, the closed-loop system \eqref{eq_cl_rob} converges to a robust positive invariant set $\mathcal{Z}_\text{\normalfont RPI}$ around the origin.
\end{theorem} 
\begin{proof}
	The proof is a straightforward adaption of \cite{koehler2017novel}, for details see \cite{hertneck2018learning}.
\end{proof}

\begin{remark}
	The size of $\mathcal{Z}_\text{\normalfont RPI}$ depends on the input error bound $\eta$. If $\eta$ is chosen small enough, $\mathcal{Z}_\text{\normalfont RPI} \subseteq \mathcal{X}_f$ and asymptotic stability of the origin can be guaranteed by applying the terminal controller in the terminal set.
\end{remark}

\section{Learning  the RMPC}
\label{sec_learning}
In this section, we discuss how the RMPC can be approximated with supervised learning methods. We also present a method for the validation of $\pi_\text{\normalfont approx}$ based on Hoeffding's Inequality to obtain guarantees on stability and constraint satisfaction.

\subsection{Supervised learning}
\label{sec:NNapprox}
To learn the RMPC, we generate an arbitrary number of samples $(x,\pi_\text{\normalfont MPC}(x)) \in \mathcal{X}_\text{\normalfont feas}\times \mathcal{U}$.  Supervised learning methods can be used to obtain the approximation $\pi_\text{\normalfont approx}$ of $\pi_\text{\normalfont MPC}$. 
To preserve guarantees from the RMPC under the approximation, we will show that 
\begin{equation}
\label{eq_approx_err}
\norm{\pi_\text{\normalfont approx}(x) - \pi_\text{\normalfont MPC}(x)}_\infty \leq \eta
\end{equation}
holds for all relevant states with the chosen $\eta$ from Assumption 3.
An approximation with a sufficient small approximation error is possible with state-of-the-art  machine learning techniques; a comprehensive overview of possible methods is given in e.g.\ \cite{shalev2014understanding}.
Equation \eqref{eq_approx_err} implies $\pi_\text{\normalfont approx}(x) = \pi_\text{\normalfont MPC} (x) + d$ with $\norm{d}_\infty \leq \eta$. Thus, Theorem \ref{theo_robust} guarantees stability and constraint satisfaction for $\pi_\text{\normalfont approx}$ if \eqref{eq_approx_err} holds.

Neural networks (NNs) are a popular and powerful method for supervised learning; see standard literature such as \cite{shalev2014understanding,GoBeCo16}. 
 Specifically, NNs have successfully been employed as control policies (see `Related work').  While we also consider NNs in the example in Section \ref{sec_example}, the proposed approach equally applies to any regression or function approximation technique  \change{}{such as \cite{chakrabarty2017support,canale2009fast}}.

It is well known that relevant classes of NNs are universal approximators; that is, they can \emph{in principle}\footnote{It can still be challenging to actually train an NN to desired accuracy in practice, e.g., because the number of required hidden units is unknown, and typically only local optima are found during training.} approximate any sufficiently regular function to arbitrary accuracy provided that the network has sufficiently many hidden units, \cite{GoBeCo16}.  However, it is in general difficult to provide an \emph{a-priori} guarantee that a learned NN satisfies the desired bound on the approximation error $\eta$. To overcome this problem, we propose a validation method based on statistical learning bounds in the next subsection.

\subsection{Probabilistic guarantees} 
In this subsection, we propose a probabilistic method to validate an approximator $\pi_\text{\normalfont approx}$. 
For the validation, we will consider trajectories of the system \eqref{eq_closed_loop}, which is controlled by the approximate MPC.  We introduce
\begin{align}
 \label{eq_def_trajectory} X_{i} := \{ &x(t), t \in \{0, \dots, T_i\} \, : \,  x(0) = x_i \in \mathcal{X}_\text{\normalfont feas}, \, \\ & x(T_i) \in \mathcal{X}_f\,   \nonumber \text{\normalfont and}~x(t+1) = f(x(t),\pi_\text{\normalfont approx}(x(t))) \}
\end{align}
to denote a trajectory of \eqref{eq_closed_loop} starting at $x(0) = x_i \in \mathcal{X_\text{\normalfont feas}}$ and ending in $\mathcal{X}_f$, where we can guarantee stability and constraint satisfaction with the terminal controller.
Then, let
\begin{equation*}
I(X_i) := 
\begin{cases}
1 & \text{\normalfont if $\norm{\pi_\text{\normalfont MPC}(x)-\pi_\text{\normalfont approx}(x)}_\infty \leq \eta, \, \forall x \in X_i$} \\
0 & \text{\normalfont otherwise}
\end{cases}
\end{equation*}
be an indicator function, which indicates whether a learned control law $\pi_{\text{\normalfont approx}}$ satisfies the posed accuracy $\eta$ along a trajectory until the terminal set is reached. 

For the validation, we consider $p$ trajectories $X_j$, $j=1,...,p$ with initial conditions $x(0)$ independently sampled from some distribution $\Omega$ over $\mathcal{X_\text{\normalfont feas}}$. 
Because the initial conditions are independent, identically distributed (iid), also $X_j$ and thus $I(X_j)$ are iid.  Next, we state a statistical bound for the approximation accuracy of $\pi_\text{approx}$ along iid trajectories \eqref{eq_def_trajectory}.

Define the empirical risk as 
\begin{equation}
\label{eq_emp_risk}
\tilde{\mu}  := \frac{1}{p} \sum_{j = 1}^{p} I(X_j).
\end{equation} 
 The RMPC guarantees stability and constraint satisfaction if
 \begin{equation}
 I(X_i)=1, \, \forall X_i \, \text{with} \, x(0)=x_i \in \mathcal{X}_\text{feas}
 \end{equation}
 holds.  The probability for $I(X_i)=1$ is $\mu := \mathbb{P} [I(X_i)=1]$ \linebreak 
for $X_i$  with 
iid initial condition $x_j(0) \in \mathcal{X}_\text{\normalfont feas}$ from the distribution $\Omega$. Thus, $\mu$ is a lower bound for the probability of stability and constraint satisfaction. We can use Hoeffding's Inequality to estimate $\mu$ from the empirical risk $\tilde{\mu}$:
\begin{lemma}
	\label{theo_hoeff}
	(Hoeffding's Inequality \cite[pp.~667-669]{von2011statistical}) \linebreak
	Let $I(X_j) ~j=1,...,p$ be $p$ iid random variables with \linebreak $0\leq I(x_j) \leq 1$.  Then,
	\begin{equation}
	\label{eq_hoeffding}
	\mathbb{P} \left[ \abs{ \tilde{\mu}- \mu}  \geq \varepsilon_h \right]  \leq   2\exp(-2p\varepsilon_h^2).
	\end{equation} 
\end{lemma}

Denote  $\delta_h := 2\exp(-2p\varepsilon_h^2)$ as the confidence level.
 Then \eqref{eq_hoeffding} implies that with confidence of at least $1-\delta_h$, 
 \begin{equation}
 \label{eq_mu_eps}
 \mathbb{P}[I(X_i)=1] = \mu \geq \tilde{\mu} - \varepsilon_h.
 \end{equation}
 Hence, with confidence $1-\delta_h$, the probability that the approximation error is below the chosen bound $\eta$ along a trajectory with a random initial condition 
 from $\Omega$ is larger than $\tilde{\mu} - \varepsilon_h$.
 We can use this 
 to establish a validation method to guarantee a chosen bound $\mu_\text{\normalfont crit}\leq \mathbb{P} \left[I(X_i) = 1\right]$  and a chosen confidence $\delta_h$. If, for a number of samples $p$, the empirical risk $\tilde{\mu}$ satisfies 
 \begin{equation}
 	\label{eq_mu_crit}
 	\mu_\text{\normalfont crit} \leq  \tilde{\mu} - \varepsilon_h = \tilde{\mu} - \sqrt{-\frac{\ln(\frac{\delta_h}{2})}{2p}},
 \end{equation} we can rewrite \eqref{eq_mu_eps} as $\mathbb{P} \left[ I(X_i) = 1 \right] \geq \mu_\text{\normalfont crit}$, which holds at least with confidence level $1-\delta_h$. We use this for validation as follows: for chosen desired confidence $\delta_h$ and $\mu_\text{\normalfont crit}$, we compute $\tilde{\mu}$ and $\varepsilon_h$  for a given number $p$ of samples. If \eqref{eq_mu_crit} holds for this $p$, the validation is successful. If \eqref{eq_mu_crit} does not hold for this $p$, the number of samples for the validation $p$ is increased, which decreases $\varepsilon_h$. The validation is then repeated iteratively while increasing  $p$. 
 We say the validation is failed, if $p$ exceeds a maximum number of samples $p_{\max}$ and stop the validation. Then the learning has to be repeated and improved (to increase $\tilde{\mu}$). \change{The proposed}{This} validation method is independent of  the chosen learning method.

Given \change{this}{the proposed} validation method, the overall procedure for the computation of an AMPC is summarized in Algorithm~1. The following theorem ensures stability and constraint satisfaction of the resulting learned AMPC. 

\begin{table}[!t]
	\newcolumntype{L}[1]{>{\raggedright\let\newline\\\arraybackslash\hspace{0pt}}m{#1}}	\vspace{0.2 cm}
\begin{tabular}{L{0.95\linewidth}}

\hline

\textbf{Algorithm 1} Learn approximate control law\\
\hline
\begin{enumerate}
\item Show incremental stabilizability. Compute $\lambda$  \newline (Assumption \ref{loc_inc_stab}, \ref{as_lip}).	
	\item Choose an accuracy $\eta$ (Assumption \ref{as_umax}).
	\item Design the RMPC: 
	\begin{enumerate}	
		\item Set $\epsilon$  according to \eqref{eq_eps_bound}.
		\item Compute terminal ingredients (Assumption \ref{as_term}).
		\item Check whether  \eqref{eq_term_inv} from Assumption \ref{as_term} holds. \newline  If not, decrease $\eta$.
	\end{enumerate}
		\item Learn $\pi_\text{\normalfont approx} \approx \pi_\text{\normalfont MPC}$, e.g. with a NN.
		\item Validate $\pi_\text{\normalfont approx}$ according to Lemma \ref{theo_hoeff}. 
		\item If the validation fails, repeat the learning from \linebreak step 4).
\end{enumerate}\\
\hline 

\end{tabular}

\end{table}

%

\begin{theorem}
	\label{lem_algo}
	Let Assumptions \ref{loc_inc_stab}, \ref{as_lip}, \ref{as_umax} and \ref{as_term} hold. Suppose that Algorithm 1 is used with suitable chosen $\eta, \mu_\text{\normalfont crit}, \delta_h$ and $p_{\max}$ and with an approximation method that can achieve an approximation error smaller than the required bound $\eta$. Then Algorithm 1 terminates. With a confidence level of $1-\delta_h$, the resulting approximate MPC ensures closed-loop stability and constraint satisfaction for a fraction of $\mu_\text{\normalfont crit}$ of the random initial conditions distributed according to $\Omega$.   	
\end{theorem}
\begin{proof}
Because Algorithm 1 terminates, \eqref{eq_mu_crit} holds. Lemma~\ref{theo_hoeff} implies \eqref{eq_mu_eps}. 
We thus have $\mathbb{P} \left[ I(X_i) = 1\right] \geq \mu_\text{\normalfont crit}$ with confidence at least $1- \delta_h$.  That is, with confidence $1-\delta_h$, for at least a fraction of $\mu_\text{\normalfont crit}$ trajectories $X_i$, we have $I(X_i) = 1$ , which implies stability and constraint satisfaction by Theorem \ref{theo_robust}.
\end{proof}

\section{Example}
\label{sec_example}
In this section, we demonstrate the AMPC scheme with a numerical example. We go step by step through Algorithm 1 and compare the resulting controller to a discrete-time Linear Quadratic Regulator (LQR). We consider a discrete-time version of a continuous stirred tank reactor with 

\begin{equation}
\nonumber
f(\tilde{x},\tilde{u}) = \begin{pmatrix}
\tilde{x}_1 + h \left(\frac{(1-\tilde{x}_1)}{\theta}-k \tilde{x}_1 e^{-\frac{M}{\tilde{x}_2}} \right) \\
\tilde{x}_2 + h\left(\frac{x_f-\tilde{x}_2}{\theta} + k \tilde{x}_1 e^{-\frac{M}{\tilde{x}_2}} -\alpha \tilde{u}(\tilde{x}_2-x_c)\right) 
\end{pmatrix},
\end{equation} where $\tilde{x}_1$ is the temperature, $\tilde{x}_2$ is the concentration, and $\tilde{u}$ is the coolant flow. This system is a common benchmark taken from\footnote{The parameters are $\theta = 20,~ k= 300,~ M = 5,~ x_f = 0.3947,~ x_c = 0.3816,~ \alpha = 0.117,~ h = 0.1$. This corresponds to an Euler discretization with the sampling time $h = 0.1s$ for the system from \cite{mayne2011tube} .} \cite{mayne2011tube}. We consider stabilization of the unstable steady state $x_e = (0.2632, 0.6519)$ with steady-state input $u_e = 0.7853$. To achieve $f(0,0) = 0$, we use the transformed coordinates $x = \tilde{x}-x_e$  and $u = \tilde{u}-u_e$. We consider the constraint sets $\mathcal{X} = [-0.2,0.2] \times [-0.2,0.2],~ \linebreak \mathcal{U} =[0-u_e,2-u_e]$, the stage cost $Q = I, R = 10^{-4}$ and use $N=180$ for the prediction horizon. 

\paragraph*{Step 1: Incremental stabilizability and Lipschitz constant}
The verification of the  local incremental stability assumption can be done according to \cite{koehler201?reference} based on a griding of the constraint set. According to this, Assumption \ref{loc_inc_stab} holds with\footnote{We use $V_\delta (x,z) = \norm{x-z}^2_{P_r}$ and $\kappa (x,z,v) = v + K_r(x-z)$ with matrices $P_r$ and $K_r$ continuous valued in $r = (x,u) \in \mathcal{X} \times \mathcal{U}$. For details on the computation and for numerical values of $P_r$ and $K_r$ see \cite{koehler201?reference,hertneck2018learning}.}  $c_{\delta,l} = 12.33,~ c_{\delta,u} = 199.03,~ k_{\max} = 45.72,~ \rho = 0.9913,~\delta_\text{\normalfont loc} = 0.01$. The Lipschitz constant is $\lambda = 5.5\cdot 10^{-3}$.

\paragraph*{Step 2: Set Accuracy}
 For the approximation accuracy $\eta$, we choose $\eta=5.1\cdot 10^{-3}$, which satisfies Assumption \ref{as_umax}. 
 \paragraph*{Step 3: RMPC design} 
 Instead of  \eqref{eq_eps_bound},  we use the following less conservative bound
 \begin{equation}
 \nonumber
 \epsilon = \eta \lambda \sqrt{c_{\delta,u}} \max \left\lbrace \tilde{k}_{\max} \norm{L_t}_{\infty}, \frac{\norm{H}_{\infty}}{\sqrt{c_{\delta,l}}} \right\rbrace =  2.2\cdot 10^{-3}	
\end{equation}  with $\tilde{k}_{\max} = \|P_r^{-1/2}K_r^\top\|$ (for details see \cite{hertneck2018learning}). Given the stage cost $Q$ and $R$, we compute the terminal cost $P$, the terminal controller $k_f(x)$ and the terminal region $\mathcal{X}_f$ based on the LQR\footnote{The terminal ingredients are 
	$k_f(x) = 
	-46.01 x_1 + 101.74 x_2,$ 
	\begin{equation*}
	P = \begin{pmatrix}
	33.21& -3.61\\
	-3.61& 6.65
	\end{pmatrix},~\alpha_f = 9.2 \cdot{10^{-5}}.
	\end{equation*}}, compare \cite{chen1998quasi}. 

 \paragraph*{Step 4: Learning}
 A  fully connected feedforward NN with two neurons in the first hidden layer and two further hidden layers with 50 neurons each is trained as the AMPC. The activations in the hidden layers are hyperbolic tangent sigmoid functions. In the output layer, linear activations are used. To create the learning samples, the RMPC is sampled over $\mathcal{X}$ with a uniform grid  using Casadi 2.4.2 \cite{Andersson2013b} with Python 2.7.2  with grid size of $2.5 \cdot 10^{-4}$. Therewith, $1.6 \cdot 10^6$ feasible samples of data points  $(x,\pi_\text{\normalfont MPC}(x))$  are generated for learning. The NN \change{}{is initialized randomly and} \change{is}{} trained with the Levenberg Marquardt algorithm 
 using Matlab's neural network toolbox on R2017a.  
 The resulting AMPC over $\mathcal{X}_\text{\normalfont feas}$ is shown in Figure~\ref{fig_surf}.
 \begin{figure}[tb] 
 	\makebox[\linewidth]{ 
 		\newlength{\figurewidth} 
 		\setlength{\figurewidth}{\linewidth} 
 		\includegraphics[height = 4.4 cm]{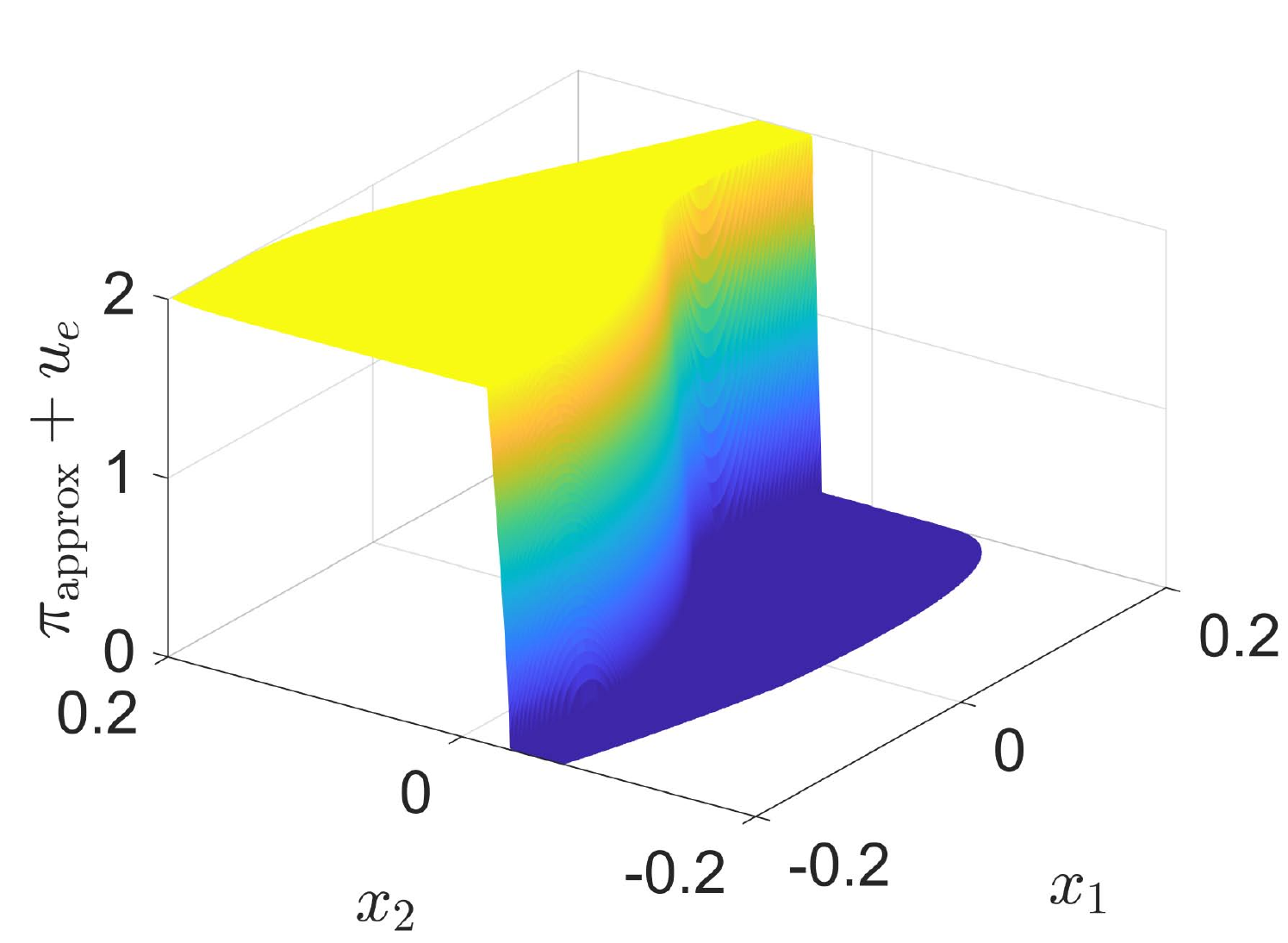}
 	}
 	\caption{Approximate MPC $\pi_\text{\normalfont approx}$ over the feasible set $\mathcal{X}_\text{\normalfont feas}$.}
 	\label{fig_surf} 
 \end{figure}
 \paragraph*{Step 5: Validation}
 We choose $\delta_h = 0.01$ and $\mu_{\text{\normalfont crit}} = 0.99$. As validation data, we sample initial conditions uniformly from $\mathcal{X}_\text{feas}$. 
 Algorithm 1 terminates, with $\tilde{\mu} = 0.9987$ for  $p= 34980$ trajectories. This implies $\mu_\text{\normalfont crit} <  \tilde{\mu} -  \varepsilon_h $  and hence, we achieve the desired guarantees. 
   We thus conclude from Theorem \ref{lem_algo} that, with a confidence of 99\%, the closed-loop system \eqref{eq_closed_loop} is stable and satisfies constraints with a probability of at least 99\%.
 
  The overall time to execute Algorithm 1 was roughly 500 hours on a Quad-Core PC. It is possible to significantly reduce this time, e.g., by parallelizing the sampling and the validation.
  
 \paragraph*{Simulation results}
 We compare the performance of the AMPC to \change{the saturated LQR controller that results from the same  stage costs}{ a standard LQR based on the linearization around the setpoint, the weights $Q,~R,$ and in addition} with $u$ saturated in $\mathcal{U}$.\footnote{\change{}{We have chosen Q and R such that the LQR shows good performance for a fair comparison. However, for some other choices (e.g. Q = I, R = 5), the LQR does not stabilize the system, while the AMPC does. 
 	}}
  Contrary to the AMPC, the LQR does not guarantee stability and constraint satisfaction for general nonlinear constrained systems. The convergence of both controllers is illustrated in Figure~\ref{fig_phase} (top) for some exemplary initial conditions. 
 
  \begin{figure}[tb]
	 
\subfigure{

 			\includegraphics[height = 5 cm, width = .9\linewidth]{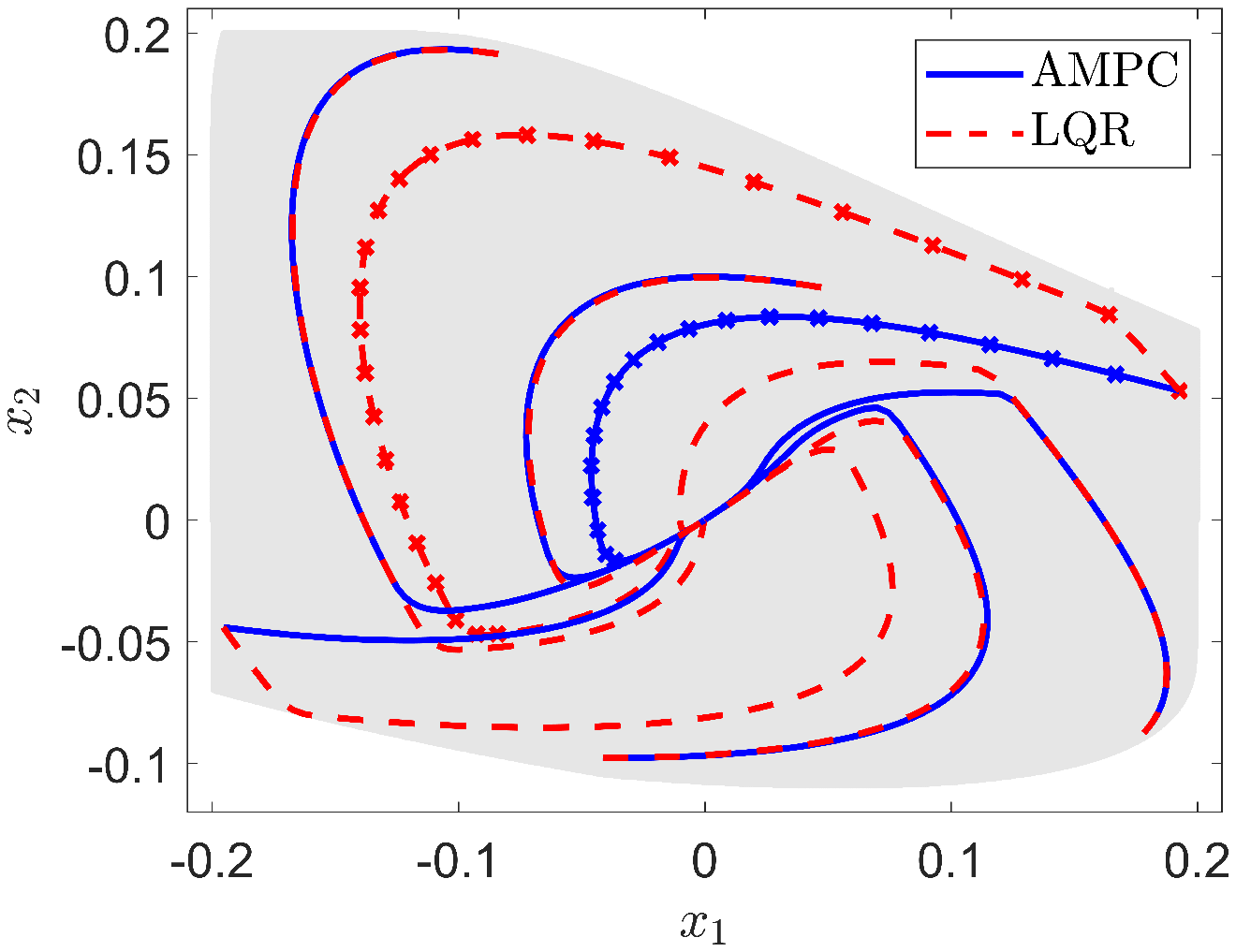} 
}

\subfigure{

	\includegraphics[width = 0.9\linewidth ]{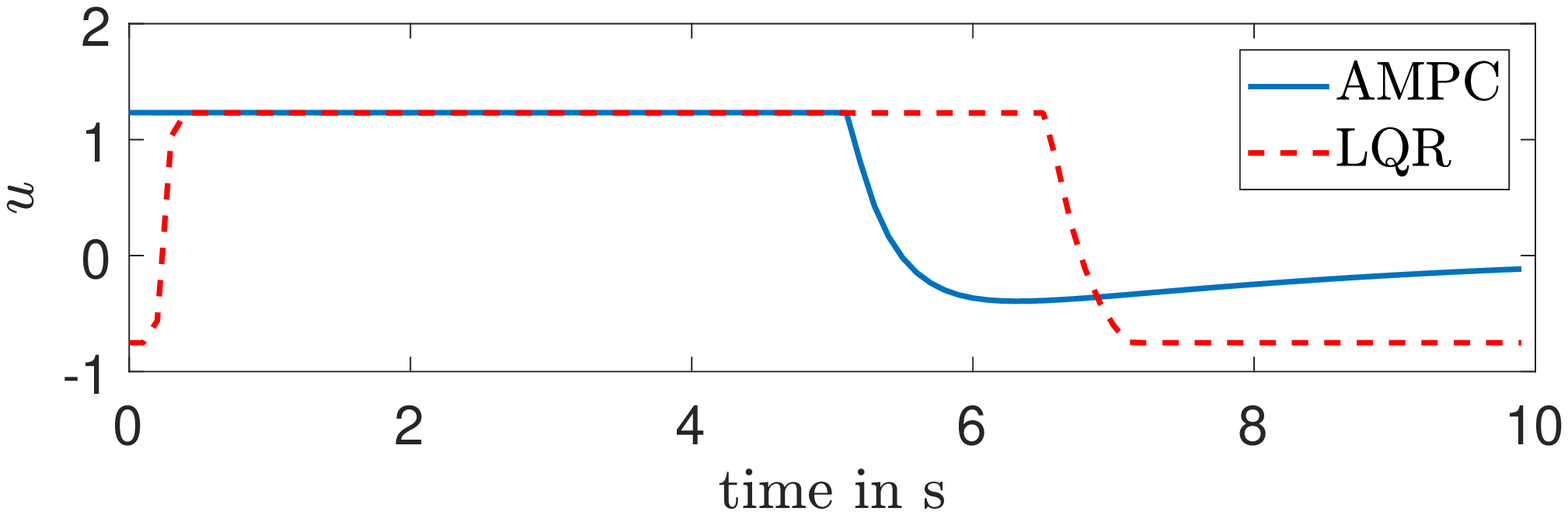}
}

 	\caption{ Top:  concentration ($x_1$) vs. temperature ($x_2$):  trajectories of AMPC (solid blue) and LQR (dashed red) and $\mathcal{X}_\text{\normalfont feas}$ (grey area). One exemplary pair of trajectories is marked with crosses.  Bottom: Coolant flow for the marked trajectory pair.}
 	\label{fig_phase} 
 \end{figure}
  While for some initial conditions, the AMPC and LQR trajectories are close (when nonlinearities play a subordinate role), there are significant differences for others. 
  For example, in Figure \ref{fig_phase} (bottom), the input flow is shown for the marked trajectory (top). Initially, the opposite input constraint is active for the LQR in comparison to the AMPC. This causes an initial divergence of the LQR trajectory leading to over three times higher costs. 
Due to the small approximation error, the trajectories of the AMPC are virtually indistinguishable from the original RMPC.
  
  \begin{remark}
  	In principle, the RMPC method from \cite{pin2013approximate} could also be used within the framework. However, the admissible approximation error for this method would be at most $\eta = 10^{-41}$, which makes the learning intractable. 
  \end{remark}
 
 To investigate the online computational demand of the AMPC, we evaluated the AMPC and the online solution of the RMPC optimization problem for 100 random points over $\mathcal{X}_\text{\normalfont feas}$.  The evaluation of {the RMPC} with Casadi took $0.71$ s on average, while the NN could be evaluated in Matlab in $3$ ms. This is already over 200 times faster \change{}{using a straightforward NN implementation in Matlab}, and further speed-up may be obtained by alternative NN implementations. 

Overall, the AMPC delivers a cost optimized, stabilizing feedback law for a nonlinear constrained system, which can be evaluated online on a cheap hardware with a high sampling rate.  
\section{Conclusion}
\label{sec_conclusion}

We proposed an RMPC scheme that guarantees stability and constraint satisfaction despite approximation errors. Stability and constraint satisfaction can thus be ensured if the approximate input is within user-defined bounds. 
Furthermore, we proposed a probabilistic method to validate the approximated MPC. By using statistical methods, we avoid conservatism, analyze complex controller structures and automate the controller synthesis with few design variables.

Tailored learning algorithms and more general verification with different indicator functions\change{}{, e.g., for the application to higher dimensional systems,} are part of future work.

\addtolength{\textheight}{-9cm}   



\section*{Acknowledgment}
The authors thank M.~Wuethrich for insightful discussions.



\bibliographystyle{IEEEtran}  
 
\bibliography{Literature}  

\end{document}